\documentclass{article}

\usepackage[left=3cm, right=3cm, bottom=4cm]{geometry}

\title{\LARGE \bf
Pricing Traffic Networks with Mixed Vehicle Autonomy}

\usepackage{epsfig} 
\usepackage{amssymb}  
\usepackage{amsmath}

\usepackage{enumitem}
\usepackage{pifont}
\usepackage[dvipsnames]{xcolor}
\usepackage{tikz}
\usepackage{pgfplots}
\usepackage{algorithm2e}
\usepackage{varwidth}
\usepackage{soul}
\usepackage{verbatim}
\usetikzlibrary{arrows,shapes}

\usepackage{graphics} 
\usepackage{amsmath} 
\usepackage{amssymb}  

\usepackage{amsthm}
\usepackage[scientific-notation=true]{siunitx}
\usepackage[colorlinks, citecolor=blue]{hyperref}
\usepackage{amsfonts}
\usepackage{tikz}
\usetikzlibrary{arrows.meta}
\usepackage{dsfont}

\usepackage{amsmath} 
\usepackage{graphicx}
\usepackage{caption}
\usepackage{subcaption}
\usepackage{multirow}

\newtheorem{theorem}{Theorem}

\newtheorem{proposition}{Proposition}
\theoremstyle{definition}
\newtheorem{example}{Example}
\newtheorem{definition}{Definition}

\theoremstyle{remark}
\newtheorem{remark}{Remark}

\usepackage{placeins}

\newcounter{tmp}

\newcommand{\cc}[1]{\mathcal{#1}}

\newcommand{\PP}{\mathcal{P}}
\newcommand{\LL}{\mathcal{L}}
\newcommand{\NN}{\mathcal{N}}
\newcommand{\WW}{\mathcal{W}}
\newcommand{\T}{\tau}

\author{Negar Mehr%
\thanks{{Negar Mehr is with the Mechanical Engineering Department, University of California, Berkeley, 5146 Etcheverry Hall, Berkeley, CA, USA
    {\tt\small negar.mehr@berkeley.edu}}}
 , and Roberto Horowitz
\thanks{Roberto Horowitz is with the Mechanical Engineering Department of University of California, Berkeley, 6143 Etcheverry Hall Berkeley, CA, USA
        {\tt\small horowitz@me.berkeley.edu}}%
}

\newcommand{\pedit}[1]{{#1}}
\newcommand{\pcomment}[1]{}

\begin{document}

\maketitle

\begin{abstract}
In a traffic network, vehicles normally select their routes selfishly.
Consequently, traffic networks normally operate at an equilibrium characterized
by Wardrop conditions. However, it is well known that equilibria are inefficient
in general. In addition to the intrinsic inefficiency of equilibria, the authors
recently showed that, \pedit{in mixed--autonomy networks in which  autonomous
  vehicles maintain a shorter headway \pcomment{but headway seems to be
    uncountable} than human--driven cars, increasing the
  fraction of autonomous vehicles in the network}
may increase the inefficiency of equilibria. In this work, we study the
possibility of obviating the inefficiency of equilibria in mixed--autonomy
traffic networks via pricing mechanisms. In particular, we study assigning
prices to network links such that  the overall or social delay of the resulting equilibria is minimum.
First, we study the possibility of inducing such optimal equilibria by imposing
a set of undifferentiated prices, i.e. a set of prices that treat both
human--driven and autonomous vehicles similarly at each link. We provide an
example which demonstrates that undifferentiated pricing is not sufficient for
achieving minimum social delay. Then, we study differentiated pricing where the
price of traversing each link may depend on whether vehicles are human--driven
or autonomous. Under differentiated pricing, we prove that link prices obtained
from the marginal cost taxation of links will induce equilibria with minimum
social delay if the degree of road capacity asymmetry (i.e. the ratio between
the road capacity when all  vehicles are human--driven and the road capacity when all vehicles are autonomous) is homogeneous among network links.


\end{abstract}

\section{Introduction}\label{sec:intro}
As autonomous vehicles become tangible technologies, studying their potential
impact on transportation networks is becoming increasingly important. With the
recent advances in deployment of autonomous vehicles, traffic networks will soon
experience a transient era when both human--driven and autonomous vehicles will
coexist on roads. Traffic networks with both vehicle types present on road, are
referred to as traffic networks with mixed vehicle autonomy. The mobility and
sustainability benefits of autonomous vehicles in networks with mixed autonomy
have been investigated from different perspectives. Multiple studies have
considered the problem of mitigating and damping shockwaves in traffic networks
with mixed vehicle
autonomy~\cite{wu2017emergent,cui2017stabilizing,stern2018dissipation} via
appropriate control of autonomous vehicles. In~\cite{sharon2017intersection}, intersection
management strategies for mixed-autonomy networks \pedit{are discussed}.
In~\cite{mehr2018game},  lane choices for autonomous vehicles \pcomment{removed
  ``in a mixed--autonomy''} \pedit{are determined} such that the overall performance of the system is optimized at traffic diverges. 


Connected and autonomous vehicles \pedit{can group into}  vehicle platoons
\pedit{that}
are capable of maintaining a shorter headway \pedit{while traveling}. 
In~\cite{lioris2017platoons}, it was shown that as a consequence of shorter headways, vehicle platooning can lead to increases in road capacities, where up to three--fold increases in road capacity were achieved when all vehicles were assumed to be autonomous. The capacity increase that results from the presence of autonomous cars in mixed--autonomy networks was modelled in~\cite{lazar2017capacity}. Throughout this paper, we assume that the headway that autonomous vehicle maintain is shorter or equal to that of human--driven vehicles and focus on the consequences of this difference.




Since vehicles select their routes selfishly, it is well known in the
transportation literature that traffic networks tend to operate at equilibria
characterized by Wardrop conditions, where vehicular flows are routed along the
network paths such that no vehicle can gain any savings in its travel time by
unilaterally changing its route~\cite{wardrop1952some}. However, it is known
that due to the selfish behavior of vehicles, such network equilibria are in
general inefficient. As an example, the well known Braess
paradox~\cite{braess1979existence} described an extreme scenario where adding a
link to a network increased the overall network delay at equilibrium.
Inefficiency of equilibria is commonly measured via the price of
anarchy~\cite{roughgarden2002bad}. In~\cite{lazar2017price}, the price of anarchy
of traffic networks with mixed vehicle autonomy was computed, and it was shown
that the price of anarchy of networks with mixed vehicle autonomy is larger than
that of networks with no autonomy. This implies that although the optimal
overall delay of networks with mixed autonomy is lower (due to the capacity
increase of autonomous vehicles), at equilibrium, the overall delay of networks
with mixed autonomy is further from its optimum. Furthermore, in our previous
work~\cite{mehr2018can, mehr2019will}, we showed that under constant vehicular
demand, if \pedit{the  fraction of  autonomous
vehicles increases},  the overall network delay may grow at equilibrium. 

In this paper, we study how to cope with the inefficiency of equilibria in traffic networks with
mixed vehicle autonomy such that the potential mobility benefits of autonomous
vehicles are   \pedit{achieved}, and minimum overall network delay \pedit{also
  known as} social delay is achieved at equilibrium.
In particular, we study how \pedit{to} set prices on network links such that equilibria with minimum social delay are induced. Pricing has been extensively studied as a tool to create
efficient equilibria in the previous literature
(See~\cite{cocchi1993pricing,cole2003pricing,cole2006much,huang2006efficiency}\pcomment{removed
outer bracket}). For traffic networks with only a single class
of vehicles, \pedit{a} marginal cost taxation of network links was proposed in~\cite{pigou2017economics} which was proven to induce equilibria with minimum social delay. However, when there
are multiple classes of vehicles, traffic networks exhibit complex behavior such as nonuniqueness of \pedit{equilibria}~\cite{dafermos1972traffic}; thereby, prices that are obtained
from marginal costs of network links may not be sufficient for inducing optimality of social delay at all possible equilibria~\cite{dafermos1973toll, smith1979marginal}.

\pedit{In this paper,} we first study whether minimum social delay can be
induced by setting undifferentiated prices on network links i.e. a set of prices
that treat both human--driven and autonomous vehicles similarly at each link. We
show through an example that undifferentiated prices are not in general enough
for inducing minimum social delay at equilibrium. Then, we consider the setting
where differentiated prices are \pedit{assigned} on network links, i.e. at every network
link, the price assigned to human--driven vehicles can be different from that of
autonomous vehicles. We prove that despite the fact that equilibrium is not
necessarily unique in traffic networks with mixed autonomy, for traffic networks
with a homogeneous degree of capacity asymmetry (\pedit{networks where the ratio
  of the link capacities when all vehicles are human--driven over that all
  vehicles are autonomous} is uniform throughout the network), with an
appropriate set of prices, the social network delay of all induced equilibria is
minimum. This is of paramount importance since uniqueness of social delay
guarantees that regardless of which equilibrium the network operates at, the
social delay will be minimum. In the absence of such a guarantee, the social
delay may be optimal only for certain  subsets of the induced equilibria.

\pedit{This paper is organized as follows.} In Section~\ref{sec:model}, we
describe how we model a traffic network with mixed autonomy. We review some
relevant results from previous works in Section~\ref{sec:prior}. In
Section~\ref{sec:impar}, we discuss the limitations of undifferentiated \pedit{pricing.}
In Section~\ref{sec:par}, we demonstrate how differentiated \pedit{pricing allows} for
inducing social optimality in certain traffic networks with mixed autonomy.
Then, in Section~\ref{sec:concl},
we conclude the paper and discuss future directions.




\section{Nonatomic Routing Games}\label{sec:model}
We model a traffic network as a directed graph $G = (\NN,\LL,\WW)$, where $\NN$ is the set of network nodes, $\LL$ is the set of network links, and $\WW$ is the set of network origin destination (O/D) pairs. For each O/D pair $w \in \WW$, $r_w^h$ denotes the fixed given demand of human--driven cars, and $r_w^a$ is the fixed given demand of autonomous cars that \pedit{need} to be routed along O/D pair $w$.
Let $r=(r_w^h,r_w^a: w \in \WW)$ be the network vector of \pedit{demands}.
We assume that the network topology is such that for each O/D pair $w \in
\WW$, there exits at least one path  connecting its origin to its destination.
We use $\PP_w$ to denote the set of all such  paths connecting \pedit{the origin
to the destination in the} O/D pair $w \in
\WW$. Moreover,  we use $\PP = \cup_{w \in \WW} \PP_w$ to denote  the set of all network paths. 

To represent flows along paths, for  each O/D pair $w \in \WW$ and path $p \in
\PP_w$, let $f_p^h$ and $f_p^a$ be  the flows of human--driven and
autonomous vehicles along path $p$\pedit{, respectively}.  Note  that each path $p \in \PP$ connects
one and only one O/D pair.
\pedit{Thus, the O/D pair associated to a path $p$ is determined unambiguously; hence, we do not carry the O/D pair associated to a path in our notation.}
\pedit{For a path $p$, $f_p = f_p^h + f_p^a$ denotes the total flow along that
  path. }
\pedit{Let $f =
(f_p^h, f_p^a : p \in \PP)$ denote the vector of human--driven  and autonomous
flows along all network paths.} A flow vector $f$ is feasible for a  given network $G$ if for every O/D pair $w \in \WW$, we have
\begin{subequations}
\label{eq:flow_cons}
\begin{align}
\sum_{p \in \PP_w} f_p^h &= r^h_w,  \\
\sum_{p \in \PP_w} f_p^a &=  r^a_w, \\
\forall p \in \PP:\; f^r_p &\geq 0, \; f^a_p \geq 0.
\end{align}
\end{subequations}


\pedit{For  link $l \in \LL$ and a flow vector $f$, we use $f_l = \sum_{p \in
    \PP: l \in p}f_p$ to denote the total flow of vehicles along link $l$.}
We further define $f_l^h = \sum_{p \in \PP: l \in p}f^h_p$ and $f_l^a = \sum_{p
  \in \PP: l \in p}f^a_p$ to   be the total flow of human--driven and autonomous
vehicles along link $l$\pedit{, respectively}.
\pedit{For a link $l \in \LL$, we denote the delay per unit of flow that is incurred when
  traveling that link by the link delay function $e_l: \mathbb{R}^2 \rightarrow \mathbb{R}$.}
For traffic networks  with mixed vehicle autonomy, using
the well known US Bureau of Public  Roads (BPR)~\cite{bureau1964traffic} delay
functions and the capacity  model of~\cite{lazar2017capacity}, the delay of
traversing each link $l\in \LL$ will be  of the following  form
(See~\cite{mehr2018can} for  details)
\begin{align}\label{eq:link_delay_fun}
e_l(f_l^h , f_l^a) = a_l+\gamma_l \left(\frac{f_l^h}{m_l}+\frac{f_l^a}{M_l}\right)^{\beta_l},
\end{align}

\noindent where $a_l$ and $\gamma_l$ are positive link parameters, $\beta_l$ is
a positive integer, $m_l$ is the capacity of link $l$ when all vehicles are
human--driven\pedit{, and} $M_l$ is the capacity of link $l$ when all vehicles are autonomous. Following~\cite{lazar2017price}, for each link $l \in \LL$, we define
\begin{align}\label{eq:deg_assym}
  \mu_l := \frac{m_l}{M_l},  
\end{align}
 to be the degree of capacity asymmetry along link $l$. Since autonomous
 vehicles are capable of maintaining a shorter headway, it is assumed that for
 every link $l\in \LL$, $M_l \geq m_l$; thus, $\mu_l \leq 1$. If the degree of
 capacity asymmetry is the same among all network links, i.e. for every link $l
 \in \LL$, $\mu_l = \mu$, then, the network is said to have a homogeneous degree
 of capacity asymmetry  $\mu$. \pcomment{removed ``equal to'' before $\mu$}

Since the delay functions are additive, we can define delay along a path $p \in \PP$ to be
\begin{align}\label{eq:path_delay}
e_p(f) := \sum_{l \in \mathcal{L}: l \in p} e_l(f_l^h, f_l^a).
\end{align}

\noindent Note that for each path $p \in \PP$, the delay along path $p$ depends
on the \pedit{whole} flow vector $f$ rather than just $f_p^h$ and $f_p^a$ since
the network links \pcomment{removed ``resources''}are shared among vehicles along different paths. 
We define the network \pedit{\emph{overall delay} (also known as \emph{social
    delay})} to be 
\begin{align}\label{eq:social_delay}
J(f) = \sum_{p \in \PP}f_p e_p(f).
\end{align}


\noindent \pedit{We say that a} flow vector $f^* = ({f_p^h}^*, {f_p^a}^*: p \in \PP)$ is socially
optimal \pedit{when} it minimizes $J(f)$ subject to
relations~\eqref{eq:flow_cons}. The optimal social delay is denoted by $J^*$. \pcomment{suggest to remove this
  sentence, as it might be confusing, prices are not yet defined}
\pedit{In this paper, we assume that a network pricing infrastructure is in
  place in which vehicles are charged as they travel along specific links.
  Moreover, vehicles may be charged differently, depending on whether they are
  human driven or autonomous.} \pedit{When the network links are priced, for each link $l \in \LL$, we use
  $\tau_l^h \geq 0$ and $\tau_l^a \geq 0$ to denote  the price for human--driven
  and autonomous vehicles along link $l$, respectively. }
Let $\tau := (\tau_l^h, \tau_l^a : l \in \LL)$ be the vector of link prices. The
price of human--driven vehicles along a path $p \in \PP$ is defined as $\tau_p^h
:= \sum_{l\in \LL:l\in p} \tau_l^h$. Likewise, define $\tau_p^a := \sum_{l\in
  \LL:l\in p} \tau_l^a$ to be the price of autonomous vehicles along path $p$.

\pcomment{what does this mean: ``such that its own delay of travel is minimized''}
A common assumption in the transportation literature is that every O/D pair
consists of infinitesimally small agents that select their routes selfishly.
\pedit{Thus,} every agent selects a route
from its origin to its destination such that its own travel delay is minimized. As a result, traffic networks achieve an equilibrium, where no agent has \pedit{an} incentive for unilaterally changing its route. A traffic network is at equilibrium if well known Wardrop conditions hold~\cite{wardrop1952some}. 

Note that when prices are set, \pedit{each agent is subjected to} both travel time and monetary costs.
Hence, for traversing a path $p \in \PP$, an agent experiences a
delay $e_p(f)$ and  pays a price equal to either $\tau^h_p$ or $\tau^a_p$,
\pedit{depending on whether it is a human--driven or an autonomous vehicle.} Thus,
assuming that all agents value travel delays and monetary costs identically, the
\emph{cost} of an agent along a  path $p$ is either $e_p(f)+\tau_p^h$ or
$e_p(f)+\tau_p^a$ \pedit{depending} on whether the agent  is human--driven or autonomous. We define the link traversal \emph{cost} functions $c_l^{h}$ and $c_l^{a}$
to be the following
\begin{subequations}\label{eq:link_cost}
\begin{align}
c_l^{h}(f_l^h, f_l^a) &:= e_l(f^h_l, f^a_l) + \tau^h_l, \\
c_l^{a}(f_l^h, f_l^a) &:= e_l(f^h_l, f^a_l) + \tau^a_l. 
\end{align}
\end{subequations}



\noindent Similarly, we define the cost of traversing a path $p \in \PP$ for human--driven and autonomous vehicles to respectively be 
\begin{subequations}
\label{eq:path-cost}
\begin{align}
c_p^{h}(f) &:= e_p(f) + \tau_p^h, \\ 
c_p^{a}(f) &:= e_p(f) + \tau_p^a.
\end{align}
\end{subequations}

\noindent For a given vector of link prices $\tau$, we define a nonatomic
selfish routing to be \pedit{the triple} $(G,r,c)$.


\begin{remark}
For every link $l \in \LL$, we use the term link \emph{cost} for human--driven
or autonomous vehicles to  refer to $c^{h}_l(f_l^r,f_l^a)$ or
$c^{a}_l(f_l^r,f_l^a)$, respectively. However, we use the term link \emph{delay} to
refer solely to $e(f_l^r,f_l^a)$, which is the delay of travel along link
$l$, \pedit{excluding the corresponding price}. Note that the cost of traversing
a link $l \in \LL$ might be different for human--driven and autonomous vehicles,
\pedit{while} the delay of traversing link $l$ is the same for both classes of vehicles. 
\end{remark}

\begin{remark}
When a price vector $\tau$ is set, although the traversal cost of a link
perceived by every agent may be different from the delay of travel along that
link, the overall performance of the system is still measured via the overall
delay incurred by all agents. The goal of this work is to find link prices such that the overall delay of the system perceived by the society is minimized.
\end{remark}


The overall \emph{cost} of a routing game $(G,r,c)$ is defined by
\begin{align}\label{eq:total-cost}
C(f) = \sum_{p \in \PP}  f_p^h c_p^{h} (f) + f_p^a c_p^{a} (f).
\end{align}

 For a priced network, Wardrop equilibria are defined via the following.

\begin{definition}
For a routing game $(G,r,c)$, a feasible flow vector $f = (f_p^h, f_p^a : p \in \PP)$ is an equilibrium if and only if for every O/D pair $w \in \WW$ and every pair of of paths $p , p' \in \PP_w$, we have
\begin{subequations}
\label{eq:eq_def}
\begin{align}
f^h_p \left(c_p^{h}(f) - c_{p'}^{h}(f)\right) &\leq 0, \label{eq:eq_def_h}\\
f^a_p \left(c_p^{a}(f) - c_{p'}^{a}(f) \right) &\leq 0.  \label{eq:eq_def_a}
\end{align}
\end{subequations}
\end{definition}

\begin{remark}
In general, despite the classical setting of a single vehicle class where \pedit{the}
Wardrop equilibrium is unique~\cite{smith1979existence}, in our mixed--autonomy
setting, there may exist multiple Wardrop equilibria
satisfying~\eqref{eq:eq_def}. \pcomment{is equilibrium unique or its cost?}
\end{remark}

\noindent Notice that equations~\eqref{eq:eq_def} imply that if for an O/D pair $w \in \WW$, and two paths $p,p' \in \PP_w$, the flows $f_p^h$ and $f_{p'}^h$ are nonzero,  we  have $c_p^{h}(f)=c_{p'}^{h}(f)$ (we can argue similarly for autonomous vehicles). Moreover, if at equilibrium, the flow along a path is zero, its travel cost cannot be smaller than that of the other paths with nonzero flow of the same vehicle class. Therefore, we can define the following.



\begin{definition}
For a routing game $(G, r, c)$, if $f$ is an equilibrium flow vector, for each
O/D pair $w \in \WW$, define the cost of travel for human--driven and autonomous
vehicles to \pedit{respectively} be
\begin{subequations}
\label{eq:eq_cost_travel}
\begin{align}
c^h_w(f) &= \min_{p \in \PP_w} c_p^{h}(f), \\
c^a_w(f) &= \min_{p \in \PP_w} c_p^{a}(f).
\end{align}
\end{subequations}
\end{definition}

\noindent Since at equilibrium, for each O/D pair \pedit{$w$ and each class of
  vehicles, the cost of travel along the paths that have nonzero flow of that
  class is the same and equal to cost of travel for that class}, we have 
\begin{align}\label{eq:social-cost-travel-cost}
C(f) = \sum_{w \in \WW}   r^h_w c_w^h (f) + r_w^a c_w^a (f) .
\end{align}



\section{Prior Work}\label{sec:prior}
In this section, we review some results from the previous literature that we will further use in this paper. The following proposition is a generalization of the results in~\cite{dafermos1973toll}.


\begin{proposition}\label{prop:marginal_price} 
For a routing game $(G,r,c)$, let $f^*$ be an optimizer of the network social delay $J$, and $J^*$ be the minimum social delay of the network. If for each link $l\in \LL$, link prices $\tau$ is set to be
\begin{subequations}
\label{eq:link-tolls}
\begin{align}
\tau_l^h &= ({f_l^h}^* + {f_l^a}^*) \left( \frac{\partial}{\partial f_l^h} e_l(f_l^h,f_l^a) \right)\bigg|_{f^*}, \\
\tau_l^a &= ({f_l^h}^* + {f_l^a}^*) \left( \frac{\partial}{\partial f_l^a} e_l(f_l^h,f_l^a) \right)\bigg|_{f^*}, \end{align}
\end{subequations}

\noindent then, there exists at least one equilibrium flow vector $f$ for the routing game $(G,r,c)$ such that the network social delay is optimal at this equilibrium, i.e. $J(f)=J^*$. 

\end{proposition}

\begin{proof}
It is easy to verify that~\eqref{eq:link-tolls} renders $f^*$ an equilibrium
flow vector by verifying \pedit{the} KKT conditions at the optimal point $f^*$. For completeness, we have included the proof of Proposition~\ref{prop:marginal_price} in Appendix~\ref{sec:appendix}.
\end{proof}

Note that Proposition~\ref{prop:marginal_price} indicates that if prices are obtained from~\eqref{eq:link-tolls}, the social delay of one of the induced equilibria is optimal. However, the social delay of other equilibria may not necessarily be optimal.

As mentioned previously, in general, the equilibrium is not unique in our mixed--autonomy setting. However, we use the following result from~\cite{altman2001equilibria} in the remainder of the paper to establish some properties of equilibria in the mixed--autonomy setting. 

\begin{proposition}\label{prop:unique}
For a routing game $(G,r,c)$, if along each link $l\in \LL$, the link traversal cost functions $c_l^h$ and $c_l^a$ are strictly increasing functions of the total flow along that link $f_l=f_l^h+f_l^a$, and the link cost functions $c_l^h$ and $c_l^a$ are identical up to additive constants, then, at equilibrium, the total flow along each link $l \in \LL$ is unique. 
\end{proposition}
It is important to mention that in our mixed--autonomy setting, since the link
cost functions~\eqref{eq:link_cost} depend on the flow of each vehicle class,
not the total flow along the link, Proposition~\ref{prop:unique} does not
\pedit{directly} apply to our setting.
\pedit{Nevertheless, we will further apply Proposition~\ref{prop:unique} to an
  auxiliary routing game to obtain some of our results. }
\section{Undifferentiated Prices}\label{sec:impar}

When it comes to setting prices for network links, generally, an ideal set of
prices is the one that induces an equilibrium flow vector that minimizes the
network social delay. When there are multiple classes of vehicles in a network,
for instance human--driven and autonomous vehicles in our \pedit{scenario}, it
is important to \pedit{determine} whether \pedit{it is possible to induce} an equilibrium that optimizes social
delay \pcomment{removed ``is achievable''} via prices that do not differentiate vehicle classes. This
is of  practical significance because undifferentiated prices are much easier to
implement.  
\pedit{Unfortunately,} in this section, we show through \pedit{a counterexample} that, for
traffic networks with mixed autonomy, it is not always possible to induce
equilibrium flows with minimum social delay by simply applying  undifferentiated prices, even in networks with a homogeneous degree of capacity asymmetry. 

\begin{figure}
\centering
\begin{tikzpicture}[scale=1.8]
\begin{scope}[every node/.style={circle,thick,draw}]
    \node (A) at (0,0) {A};
    \node (B) at (-2,1) {B};
    \node (C) at (-2,-1) {C};
    
\end{scope}

\begin{scope}[>={Stealth[teal]},
              every node/.style={fill=white,circle},
              every edge/.style={draw=teal,very thick}]
    \path [->] (A) edge node {$1$} (B);
    \path [->] (A) edge node {$2$} (C);
    \path [->] (B) edge node {$3$} (C);
    \path [->] (C) edge[bend left=45] node {$4$} (B); 
\end{scope}

\end{tikzpicture}
\caption{A network with two O/D pairs from A to B and A to C.}
\label{fig:couter_examp_net}
\end{figure}
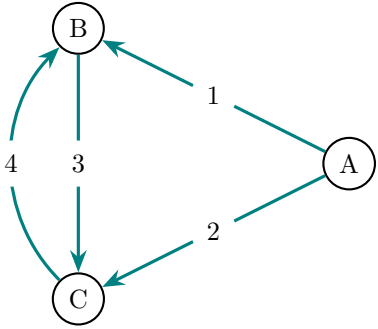
\begin{example}\label{examp:counter_examp}
Consider the network shown in Figure~\ref{fig:couter_examp_net}. Assume that the
network has a homogeneous degree of capacity asymmetry $\mu = \frac{1}{3}$.
There are two O/D pairs $\WW = \{ AB,AC\}$. For each O/D pair, there are two
possible paths. The demand of O/D pairs are $r^h_{AB}= 7.5$, $r^a_{AB}= 4.5$,
$r^h_{AC}=1.2$, and $r^a_{AC}= 4.8$. For every link $l, 1 \leq l \leq 4$, assume
that $\gamma_l = 1$, and $\beta_l = 1$. The other link parameters are $a_1 = 9,
a_2 = 3, a_3 = 0.6$, and $a_4 = 0.6$ while $m_1 = 3, m_2 = 0.5, m_3 = 0.7$, and
$m_4 = 0.5$. \pedit{Note that for each link $l$, the parameter $M_l$ is determined through the
  relation $M_l = m_l / \mu$.} For this network, \pedit{it is easy to compute}
the optimal social delay\pedit{, which }is $J^*= 193.54$.
\pedit{In order to see whether a set of undifferentiated prices can achieve this
optimal social delay,}
we  \pedit{solve} the following optimization problem for this network
\begin{equation}\label{eq:examp_opt}
\begin{aligned}
& \underset{}{\min}
& & \sum_{p \in \PP} (f_p^h + f_p^a) e_p(f) \\
& \text{subject to}
& & \text{Equations}~\eqref{eq:flow_cons}, \\
&&&  \text{Equations}~\eqref{eq:eq_def},\\
&&& \forall l \in \LL: \T_l^h = \T_l^a \geq 0,
\end{aligned}
\end{equation}
\pcomment{why not bringing the constraint $\tau>0$ downstairs, e.g. $\forall l
  \in \LL: \T_l^h = \T_l^a \geq 0$}
Note that if \pedit{the} minimum social \pedit{delay $J^*$ can be achieved}  via undifferentiated prices, the optimal value of optimization problem~\eqref{eq:examp_opt} must be equal to $J^*$. However, the minimum value of~\eqref{eq:examp_opt} for the network in Figure~\ref{fig:couter_examp_net} is $195.597$ which is clearly greater than $J^*$. This indicates that an equilibrium with socially optimal delay cannot necessarily be induced by undifferentiated link prices in general.
\end{example}
\section{Differentiated prices}\label{sec:par}

Having shown in Example~\ref{examp:counter_examp}
that in general, undifferentiated pricing \pedit{cannot induce} an
equilibrium with socially optimal delay, a natural question \pcomment{removed
  ``that arises''} is whether \emph{differentiated} prices can be employed to
induce an equilibrium with minimum social delay. In other words, if at every
link $l \in \LL$, we allow $\tau_l^h$ and $\tau_l^a$ to be different, does there
exist a price vector $\tau = (\tau_l^h, \tau_l^a : l \in \LL)$ that induces
equilibria with minimum social delay? \pcomment{suggest not to emph delay}
In this section, \pedit{we prove that such differentiated prices
exist, and we find them.}

\subsection{Homogeneous Networks}
The following theorem establishes the existence of optimal prices for traffic networks with a homogeneous degree of capacity asymmetry and also provides a recipe for how to find the optimal price values.

\begin{theorem}\label{theo}
Consider a routing game $ (G,r,c)$ with a homogeneous degree of capacity
asymmetry $\mu$.
\pedit{Let $f^*$ and $J^*$ be an optimizer and  the minimum value of the social
  delay, respectively. 
}
Then, for each link $l \in \LL$, if \pedit{the} link prices are set to be
\begin{subequations}
\label{eq:link-prices-theor}
\begin{align}
\tau_l^h &= ({f_l^h}^* + {f_l^a}^*) \left( \frac{\partial}{\partial f_l^h} e_l(f_l^h,f_l^a) \right)\bigg|_{f^*}, \\
\tau_l^a &= ({f_l^h}^* + {f_l^a}^*) \left( \frac{\partial}{\partial f_l^a} e_l(f_l^h,f_l^a) \right)\bigg|_{f^*}, \end{align}
\end{subequations}
then,
\pedit{all induced equilibria of the game $(G,r,c)$ have the same social delay,
  which is equal to $J^*$.}
\end{theorem}
\begin{proof}

For \pedit{the} routing game $(G,r,c)$ where \pedit{the} link prices are
obtained via~\eqref{eq:link-prices-theor},
\pedit{Proposition~\ref{prop:marginal_price} implies that there exists one
  equilibrium with minimum social delay $J^*$.}
We prove that when \pedit{the above} prices are set, all induced equilibria of
$(G,r,c)$ have the same social delay.  \pedit{This would then imply that all induced equilibria of
  the game $(G, r, c)$} have the unique social delay $J^*$, as was claimed. 



\pedit{Therefore, it remains to prove uniqueness of social delay at equilibria of $(G,r,c)$ when link prices
are obtained from~\eqref{eq:link-prices-theor}. In order to do so,} we construct an auxiliary game
instance $(G,\tilde{r},\tilde{c})$, \pedit{with the same network graph $G$ and
  O/D pairs $\WW$,} where the demand of O/D pairs, link traversal delays and
cost functions are defined as follows. For each O/D pair $w \in \WW$,
define the demand of human--driven and autonomous cars $\tilde{r}^h_w$ and
$\tilde{r}^a_w$ in the auxiliary game to  be
\begin{subequations}\label{eq:r_tilde}
\begin{align}
   \tilde{r}^h_w &:= r^h_w, \\
   \tilde{r}^a_w &:= \mu r^a_w.
\end{align}
\end{subequations}

\noindent Moreover, for every link $l \in \LL$, let the link delay functions of
the auxiliary game $(G,\tilde{r},\tilde{c})$ be \pedit{defined as}
\begin{align}\label{eq:e_tilde}
    \tilde{e}_l (\tilde{f}_l^h,\tilde{f}_l^a) := a_l+\gamma_l \left(\frac{\tilde{f}_l^h+\tilde{f}_l^a}{m_l}\right)^{\beta_l}.
\end{align}

\noindent Additionally, for every link $l\in \LL$, \pedit{with the prices as in
  \eqref{eq:link-prices-theor},} define the link cost functions in the auxiliary game to be
\begin{subequations}\label{eq:c_tilde}
\begin{align}
   \tilde{c}^{h}_l(\tilde{f}_l^h,\tilde{f}_l^a) &:= \tilde{e}_l (\tilde{f}_l^h,\tilde{f}_l^a) + \T_l^h, \\
   \tilde{c}^{a}_l(\tilde{f}_l^h,\tilde{f}_l^a) &:= \tilde{e}_l (\tilde{f}_l^h,\tilde{f}_l^a) + \T_l^a.
\end{align}
\end{subequations}

Now, let $f=(f_p^h,f_p^a: p \in \PP)$ be an equilibrium of the original game
$(G,r,c)$. For every path $p \in \PP$, define $\tilde{f}_p^h := {f}_p^h$ and
$\tilde{f}_p^a := \mu {f}_p^a$. We claim that $\tilde{f} = (\tilde{f}_p^h,
\tilde{f}_p^a : p \in \PP)$ is an equilibrium flow vector for the auxiliary game
$(G,\tilde{r},\tilde{c})$.  \pcomment{since auxiliary flows are defined through
  paths, they are always admissible. so removed ``Note that since the network is
  assumed to have a homogeneous degree of capacity asymmetry, $\tilde{f}$ is an
  admissible flow vector that satisfies flow conversation at every node.''} It
can be easily verified that for every origin--destination pair $w \in \cc{W}$,
we have $\sum_{p \in \PP_w}\tilde{f}_p^h = \tilde{r}_w^h$ and  $\sum_{p \in
  \PP_w}\tilde{f}_p^a = \tilde{r}_w^a$. Thus, $\tilde{f}$ is a feasible flow
vector for the auxiliary game.  \pedit{Moreover, it is easy to see that for
  every link $l \in \LL$, we have $\tilde{f}_l^a = \mu f_l^a$; therefore, using}
the definition of $\tilde{f}$ and Equations~\eqref{eq:link_delay_fun}
and~\eqref{eq:e_tilde}, for every link $l\in \LL$, we establish the following
\begin{equation}
  \label{eq:equal-delay}
\begin{aligned}
\tilde{e}_l (\tilde{f}_l^h,\tilde{f}_l^a) &= a_l+\gamma_l \left(\frac{\tilde{f}_l^h+\tilde{f}_l^a}{m_l}\right)^{\beta_l}\\
&= a_l+\gamma_l\left(\frac{f_l^h+\mu f_l^a}{m_l}\right)^{\beta_l}\\
&= a_l+\gamma_l \left(\frac{f_l^h}{m_l}+\frac{f_l^a}{M_l}\right)^{\beta_l}\\
&= e_l(f_l^h,f_l^a).
\end{aligned}
\end{equation}
\pcomment{replaced only one equation number for the whole chain.}

\noindent Thus, from~\eqref{eq:link_cost}, \eqref{eq:c_tilde},
 and~\eqref{eq:equal-delay}, for every link $l\in \LL$, we have
\begin{subequations}
  \label{eq:cost_equi}
\begin{align}
   \tilde{c}^{h}_l(\tilde{f}_l^h,\tilde{f}_l^a) &= c^{h}_l({f}_l^h,{f}_l^a) \\
   \tilde{c}^{a}_l(\tilde{f}_l^h,\tilde{f}_l^a) &= c^{a}_l({f}_l^h,{f}_l^a).
\end{align}
\end{subequations}

Now since $f$ is an equilibrium for $(G,r,c)$, using~\eqref{eq:eq_def} and~\eqref{eq:cost_equi}, for every O/D pair $w \in \WW$ and pair of paths $p, p' \in \PP_w$, we have
\begin{subequations}
\label{eq:eq_def_repeat_rep}
\begin{align}
f^h_p \left(\tilde{c}_p^{h}(\tilde{f}) - \tilde{c}_{p'}^{h}(\tilde{f})\right) &\leq 0, \\
f^a_p \left(\tilde{c}_p^{a}(\tilde{f}) - \tilde{c}_{p'}^{a}(\tilde{f}) \right) &\leq 0. \label{eq:second_eq_1}
\end{align}
\end{subequations}

\noindent Multiplying~\eqref{eq:second_eq_1} by the positive constant $\mu$, we have
\begin{subequations}\label{eq:tilde_f_eq}
\begin{align}
f^h_p \left(\tilde{c}_p^{h}(\tilde{f}) - \tilde{c}_{p'}^{h}(\tilde{f})\right) &\leq 0, \\
\mu f^a_p \left(\tilde{c}_p^{a}(\tilde{f}) - \tilde{c}_{p'}^{a}(\tilde{f}) \right) &\leq 0.
\end{align}
\end{subequations}

\noindent Using the definition of $\tilde{f}$, from~\eqref{eq:tilde_f_eq} we can conclude the following
\begin{subequations}
\label{eq:tilde_f_equ}
\begin{align}
\tilde{f}^h_p \left(\tilde{c}_p^{h}(f) - \tilde{c}_{p'}^{h}(f)\right) &\leq 0, \\ \tilde{f}^a_p \left(\tilde{c}_p^{a}(\tilde{f}) - \tilde{c}_{p'}^{a}(\tilde{f}) \right) &\leq 0. \label{eq:second_eq}
\end{align}
\end{subequations}
\pedit{Note that these are precisely the equilibrium conditions for the
  auxiliary game $(G,\tilde{r},\tilde{c})$, which proves our claim that
  $\tilde{f}$ is an equilibrium for the auxiliary game.}

Now, note that the conditions of Proposition~\ref{prop:unique} hold for the
auxiliary game $(G,\tilde{r},\tilde{c})$ since for every link $l\in \LL$, the
link traversal costs of human--driven and autonomous cars $\tilde{c}_l^h$ and
$\tilde{c}_l^a$ are strictly increasing functions of the total link flow
$\tilde{f}_l = \tilde{f}_l^h+\tilde{f}_l^a$. Moreover, \pedit{motivated by
  \eqref{eq:c_tilde},} the costs of human-driven
and autonomous cars are identical up to a constant. Thus, using
Proposition~\ref{prop:unique}  for $(G,\tilde{r},\tilde{c})$, for every link $l \in \LL$,  the total link flow
$\tilde{f}_l = \tilde{f}_l^h+\tilde{f}_l^a$ is unique \pedit{among all the equilibria}. Therefore, using the
definition of $\tilde{f}$, the fact that $\tilde{f}$ is an equilibrium flow
vector for $(G,\tilde{r},\tilde{c})$, and the connection between $f$ and
$\tilde{f}$, we conclude that at every link $l \in \cc{L}$, we must have that
$f_l^h+ \mu f_l^a$ is unique for all equilibria of $(G,r,c)$. Additionally,
from~\eqref{eq:e_tilde} and~\eqref{eq:c_tilde}, for every link $l \in \LL$,
uniqueness of \pedit{the} total link flow at equilibrium in the auxiliary game implies that
the link traversal costs ${\tilde{c}}^h_l (\tilde{f}_l^h,\tilde{f}_l^a)$ and
${\tilde{c}}^a_l(\tilde{f}_l^h,\tilde{f}_l^a)$  are unique. Hence,
from~\eqref{eq:cost_equi}, we can conclude that in $(G,r,c)$, for each link $l
\in \LL$, the link traversal costs $c_l^h(f_l^h,f_l^a)$ and $c_l^a(f_l^h,f_l^a)$
are also unique for all \pcomment{removed ``Wardrop''} equilibrium flow vectors $f$. Thus, for each O/D
pair $w \in \WW$,~\eqref{eq:eq_cost_travel} results in uniqueness of travel
costs of both human-driven and autonomous cars $c_w^h$ and $c_w^a$.
Consequently, from~\eqref{eq:social-cost-travel-cost}, we realize that the
overall cost $C(f)$  is unique for all equilibrium flow vectors $f$ of $(G,r,c)$.

 Now, using~\eqref{eq:path_delay},~\eqref{eq:link_cost}, and~\eqref{eq:total-cost}, we can rewrite the social cost of $(G,r,c)$ as
\begin{align}
    C(f) &= \sum_{l \in \LL}  f_l^h c_l^{h} (f) + f_l^a c_l^{a} (f) \\
    &= \sum_{l \in \LL}  (f_l^h+f_l^a) e_l^{h} (f)+ f_l^h \tau_l^h + f_l^a \tau_l^a \\
    &= J(f) + \sum_{l\in \LL} f_l^h \tau_l^h + f_l^a \tau_l^a. \label{eq:soc_cost_struc}
\end{align}
Notice that under homogeneity of the network, using the special structure
of~\eqref{eq:link_delay_fun},  \pedit{it is easy to see that} Equation~\eqref{eq:link-prices-theor} implies that for every link $l \in \LL$, we have
\begin{align}\label{eq:price-struc}
    \tau_l^a = \mu \tau_l^h.
\end{align}
Substituting~\eqref{eq:price-struc} in~\eqref{eq:soc_cost_struc}, we have
\begin{align}\label{eq:social-cost-simp}
    C(f) = J(f) + \sum_{l\in \LL} \tau_h^l (f_h^l + \mu f_l^a). 
\end{align}
Note that using our introduced auxiliary game, we proved that the overall cost
$C(f)$ is unique for all \pcomment{removed ``Wardrop''} equilibrium flow vectors $f$. Furthermore,
we proved that for every link $l$, \pedit{$f_l^h+\mu f_l^a = \tilde{f}_l^h + \tilde{f}_l^a$} is unique for all \pcomment{removed ``Wardrop''} equilibria. Hence, from~\eqref{eq:social-cost-simp}, we can conclude that the social delay $J(f)$ is also unique for all equilibrium flow vectors. This completes our proof.
\end{proof}

Note that if for each link $l\in \LL$, the link prices are obtained from~\eqref{eq:link-prices-theor}, since the link delay function~\eqref{eq:link_delay_fun} is increasing in the flow of each vehicle class, the prices that result from~\eqref{eq:link_delay_fun} are always nonnegative which is in accordance with our initial assumption. The link prices obtained by~\eqref{eq:link-prices-theor} are in fact the extra term in the marginal cost of each vehicle class. 


\subsection{Heterogeneous Networks}
\pcomment{one idea is to remove this short subsection?}
    



If the road degree of capacity asymmetry is not homogeneous, \pedit{but} a
central authority sets link prices to be obtained
from~\eqref{eq:link-prices-theor}, there still exists at least one induced
equilibrium flow vector \pedit{that achieves the} minimum social delay. However, for heterogeneous networks, the social delay is not necessarily unique. Therefore, although the social delay of one induced equilibrium is optimal, the social delay of other induced equilibria might not be optimal. 
For such networks, optimally of the social delay may not be achieved in all induced equilibria by setting the prices to be obtained from~\eqref{eq:link-prices-theor}.




\label{sec:pricing}
\section{Conclusion and Future Work}
We considered the problem of inducing efficient equilibria in traffic networks with mixed vehicle autonomy via pricing. We showed that minimum social delay may not be attained by imposing undifferentiated link prices, in which human--driven and autonomous vehicles are treated identically.  Then, we proved that in mixed--autonomy traffic networks with a homogeneous degree of capacity asymmetry, if differentiated prices are allowed, which treat human--driven and autonomous vehicles differently, link prices can be determined such that all induced equilibria have minimum social delay. For future steps, it is interesting to study path--based price collection to achieve further objectives such as collecting the minimum possible monetary value or obtaining a fair pricing policy. It is also important to study the existence of prices when users are heterogeneous, i.e. users that might value the monetary costs of prices differently. \label{sec:concl}
\section*{Acknowledgments}
This work is supported by the National Science Foundation under Grant CPS 1545116.

\appendix

\section{Proof of Proposition~\ref{prop:marginal_price}}\label{sec:appendix}
Since the social delay defined by~\eqref{eq:social_delay} is a continuous function of flows along network paths, and the set of feasible flows satisfying~\eqref{eq:flow_cons} is compact, there exists a flow vector $f^*$ that optimizes the social delay, i.e. $f^*$ is the optimizer of the following optimization problem

\begin{equation}\label{eq:social-delay-optimizer}
\begin{aligned}
& \underset{f}{\min}
& & J(f) \\
& \text{subject to}
& & \forall p \in \PP: f_p^h \geq 0,\, f_p^a \geq 0, \\
 &&&  \forall w \in \WW: \sum_{p \in \PP_w} f_p^h = r_w^h, \sum_{p \in \PP_w} f_p^a = r_w^a.
\end{aligned}
\end{equation}


Note that all constraints in~\eqref{eq:social-delay-optimizer} are affine
functions of the decision variable $f$. Therefore, the regularity conditions
hold for~\eqref{eq:social-delay-optimizer} \pedit{(see, for instance,
  Theorem~5.1.3 and Lemma~5.1.4 in \cite{bazaraa2013nonlinear}).}  Thus, the
optimizer $f^*$ must satisfy \pedit{the} KKT  conditions. For each path $p \in \PP$, let
$\lambda_p^h \geq 0$ and $\lambda_p^a \geq 0$  be the Lagrange
multipliers associated with the nonnegativity constraint of the flows of
human--driven and autonomous vehicles along path $p$, respectively. Similarly, for  each O/D
pair $w\in \WW$,
let $\nu_w^h$ and $\nu_w^a$  be the Lagrange multipliers associated with
\pedit{the} flow
conservation constraints for human--driven and autonomous vehicles along the O/D pair
$w$, respectively. Then, for a fixed path $p \in \PP_w$ \pedit{associated to} an O/D pair $w \in \WW$, for the flow of human--driven cars, the stationarity condition imposes the following
\begin{align}\label{eq:stationariy}
\frac{\partial}{\partial f_p^h} J(f)\bigg|_{f^*} &= \lambda_p^h - \nu_w^h.  
\end{align}
From~\eqref{eq:path_delay},~\eqref{eq:social_delay},
and~\eqref{eq:link-prices-theor}, we have \pcomment{note subscripts should be
  $l$ not $p$ on the second line}
\begin{align*}
\begin{split}
\frac{\partial}{\partial f_p^h} J(f)\bigg|_{f^*} &= \\ \sum_{l\in \LL: l\in p} \Big( & e_l(f_l^h, f_l^a) + (f_{\pedit{l}}^h + f_{\pedit{l}}^a) \frac{\partial}{\partial f_{\pedit{l}}^h} e_l(f_l^h,f_l^a) \Big) \bigg|_{f^*} 
\end{split}\\
&= e_p(f^*)+\tau_p^h.
\end{align*}
Using this together with~\eqref{eq:stationariy}, we can conclude that for every path $p \in \PP$, we have
\begin{align}\label{eq:stationarity-simpl}
    e_p(f^*)+\tau_p^h = \lambda_p^h - \nu_w^h.
\end{align}
On the other hand, complementary slackness requires that for every path $p\in\PP$
\begin{align}\label{eq:compl-slack}
    \lambda_p^h f_p^h = 0.
\end{align}
Now, for a fixed O/D pair $w \in \WW$, consider a pair of paths $p,p'\in \PP_w$. If $f_p^h > 0$, from~\eqref{eq:compl-slack}, we must have $\lambda_p^h = 0$. Then, from~\eqref{eq:stationarity-simpl} and nonnegativity of $\lambda_{p'}^h$, we have
\begin{align}\label{eq:equ-cond-hum}
e_p(f^*)+\tau_p^h = -\nu^h_w \leq \lambda^h_{p'}-\nu^h_w = e_{p'}(f^*)+\tau^h_{p'},
\end{align}
\pedit{where in the last equality, we have used~\eqref{eq:stationarity-simpl}
  for the path $p'$.}
Similarly, for autonomous cars along the two paths $p$ and $p'$, if $f_p^a > 0$, we must have
\begin{align}\label{eq:equ-cond-aut}
e_p(f^*)+\tau_p^a = -\nu^a_w \leq \lambda^a_{p'}-\nu^a_w = e_{p'}(f^*)+\tau^a_{p'}.
\end{align}
\pedit{Note that \eqref{eq:equ-cond-hum} implies~\eqref{eq:eq_def_h}, the reason
being that if $f_p^h=0$, \eqref{eq:eq_def_h} automatically holds, and if $f_p^h
> 0$, \eqref{eq:eq_def_h} holds since as we showed above, $c_p^h(f^*) \leq
c_{p'}^h(f^*)$. Likewise, \eqref{eq:equ-cond-aut} implies \eqref{eq:eq_def_a}.}
Hence, once prices are set according to~\eqref{eq:link-prices-theor}, the optimal flow $f^*$ is a Wardrop equilibrium for the game $(G,r,c)$; thus, there exists at least one induced equilibirum with minimum social delay once prices are obtained from~\eqref{eq:link-prices-theor}.










\end{document}